\documentclass{article}
\usepackage{amsmath,amsfonts,amssymb,amsthm,bbm,setspace}
\usepackage{bbm}
\usepackage{changepage}
\usepackage{algorithm} 
\usepackage{algpseudocode} 
\usepackage{mathtools}
\usepackage{commath}
\usepackage{graphicx}
\usepackage{xcolor}
\usepackage{setspace}
\usepackage{hyperref}
\usepackage[margin = 1 in]{geometry}
\usepackage{multirow}

\newtheorem{theorem}{Theorem}
\newtheorem{lemma}{Lemma}

\newcommand{\E}{\mathbb{E}}
\newcommand{\Ber}{\mathrm{Ber}}
\newcommand{\supp}{\mathrm{supp}}

\title{The Noisy Quantitative Group Testing Problem}
\author{Tenghao Li \and Neha Sangwan \and Xiaxin Li \and Arya Mazumdar
\thanks{The authors are all associated with University of California, San Diego. 
}}
\date{}
 
\begin{document}
\maketitle
\begin{abstract}
In this paper, we study the problem of quantitative group testing (QGT) and analyze the performance of three models: the noiseless model, the additive Gaussian noise model, and the noisy Z-channel model.
For each model, we analyze two algorithmic approaches: a linear estimator based on correlation scores, and a least squares estimator (LSE).
We derive upper bounds on the number of tests required for exact recovery with vanishing error probability, and complement these results with information-theoretic lower bounds.
In the additive Gaussian noise setting, our lower and upper bounds match in order.
\end{abstract}

\section{Introduction}
Quantitative group testing (QGT)\cite{feige_quantitative_2020}\cite{karimi2019} generalizes classical group testing by observing, for each pooled test, the \emph{number} of defectives in the pool.
Compared with binary outcomes, quantitative measurements can substantially reduce the number of tests required for exact support recovery and can enable simple and efficient decoding algorithms.
In this work, we study QGT under a non-adaptive random testing design in the sublinear regime, where the number of defectives grows sublinearly with the total number of items, and we focus on three canonical observation models.
\subsection{Noiseless Model}
In the \emph{noiseless} QGT model, each test returns the exact number of defective items in the pool. This setting serves as a clean baseline in the absence of uncertainty. The noiseless model has been studied extensively in the literature\cite{alaoui2014}\cite{hahn-klimroth_efficient_2023}\cite{karimi_sparse_2019}. In the linear regime, where the fraction of defectives is constant, information-theoretic thresholds were characterized by Alaoui \emph{et al.}~\cite{alaoui2014} and Scarlett and Cevher~\cite{scarlett2017phase}. 

In the remainder of this paper, we focus on the \emph{sub-linear regime}, where the number of defectives $k=n^\theta$ for some fixed $\theta\in(0,1)$.  In the noiseless setting, it is widely accepted that $\Theta(k\log(n/k)/\log k)$ non-adaptive tests are both necessary and sufficient for exact recovery~\cite{djackov1975}~\cite{hahn-klimroth_near_2022}. Moreover, under the scaling $k=n^\theta$, we have
$\frac{\log(n/k)}{\log k}=\frac{1-\theta}{\theta}$,
which is a constant depending only on $\theta$. Thus, in the sub-linear regime, the optimal sample complexity remains linear in $k$ up to a constant factor determined by the sparsity exponent. This is qualitatively different from the linear regime, where $k=\alpha n$ for some constant $\alpha\in(0,1)$ and the relevant sample complexity scales with $n$.

On the algorithmic side, a variety of non-adaptive decoding methods have been proposed for noiseless QGT. Gebhard \emph{et al.}~\cite{gebhard2019quantitative} introduced a greedy reconstruction algorithm, while coding-theoretic constructions based on sparse graph codes were developed by Karimi \emph{et al.}~\cite{karimi2019}. Feige and Lellouche~\cite{feige_quantitative_2020} further showed that many polynomial-time algorithms require significantly more tests than the information-theoretic optimum, and demonstrated improved recovery guarantees via reductions to structured subset selection problems. More recently, Hahn-Klimroth \emph{et al.}~\cite{hahn-klimroth_near_2022} proposed a polynomial-time algorithm achieving a near-optimal number of tests in the sub-linear regime.

\subsection{Additive Gaussian Noise Model}
In many practical settings, pooled measurements are corrupted by continuous noise, motivating the \emph{additive Gaussian noise} QGT model~\cite{hahn-klimroth_near_2022,tan2024approximate}. In this model, each test outcome is perturbed by independent Gaussian noise. Dedicated studies of Gaussian noise in QGT have only recently begun to emerge.

Scarlett and Cevher~\cite{scarlett2017phase} established converse bounds for more general pooled-data problems with multi-category observations, where each test reveals the number of items of each type present in the pool, without imposing a sparsity assumption on the underlying signal. In contrast, we provide a simpler converse argument specialized to the sub-linear QGT regime, and complement it with a matching achievability bound. To the best of our knowledge, this gives the first tight characterization of the sample complexity of Gaussian QGT up to constant factors.

On the algorithmic side, Hahn-Klimroth and Kaaser~\cite{hahn-klimroth_distributed_2022} analyzed a greedy correlation-based decoder and showed that exact recovery is achievable with high probability using $O(k\ln n)$ tests when the Gaussian noise variance satisfies $\sigma^2=o(m/\ln n)$. We also study a related linear estimator, but our analysis removes this restriction on $\sigma^2$ by combining sharper subgaussian tail bounds with a more careful decomposition of the correlation scores.

\subsection{Noisy Z-Channel}
The \emph{noisy Z-channel} model captures asymmetric item-level unreliability. In this setting, each defective item included in a test may independently fail to contribute to the observed count, leading to false negatives but no false positives.

Hahn-Klimroth and Kaaser~\cite{hahn-klimroth_distributed_2022} studied a related noisy channel model in which each item flips independently prior to pooling. Their framework includes symmetric flip noise as well as one-sided error models as special cases, and they showed that exact recovery is achievable with high probability using $O(k\ln n)$ tests. We study a related linear estimator and obtain a similar test complexity. Furthermore, in the limit where the flip probability vanishes, our bound recovers the corresponding noiseless linear-estimator guarantee. We also derive achievability and converse bounds for the noisy Z-channel model.

Techniques from noisy binary group testing also provide useful insights. For example, Price \emph{et al.}~\cite{price2022} developed fast splitting algorithms with near-optimal test complexity and runtime under stochastic test errors, though their analysis focuses on Boolean outcomes.

\subsection{Other Noise Models}

Beyond the mentioned models, a few works have explored adversarial or structured corruption models.
In particular, Li and Wang~\cite{li_wang_2021} studied combinatorial QGT under bounded adversarial noise, where test outcomes may be perturbed arbitrarily within fixed limits.
In this setting, exact recovery is generally impossible, and the goal shifts to partial recovery.
They established lower bounds on the order of $\Omega(n\log n)$ tests and proposed deterministic test designs with efficient decoding that achieve these bounds up to constant factors.

\subsection{Our Contributions}
Table~\ref{tab:result summary} compares our results with existing work.
Our main contributions are summarized as follows:
\begin{itemize}
    \item 
    We provide a unified analysis of three quantitative group testing models: noiseless, additive Gaussian noise, and noisy Z-channel.
    This contrasts with prior work, which typically focuses on individual noise models or specific problem formulations.
    \item
    We provide both achievability and converse bounds for the two noisy models.
    In particular, for the additive Gaussian noise model, we establish matching achievability and converse bounds on the number of tests required for exact recovery.
    To the best of our knowledge, this yields the first tight characterization of the sample complexity for Gaussian QGT up to constant factors.
    \item
    We analyze two decoding approaches across all models:
    a computationally efficient correlation-based linear estimator and a statistically optimal (but generally combinatorial) least squares estimator (LSE).
    This allows us to explicitly quantify the performance gap between polynomial-time algorithms and information-theoretically optimal decoding under different noise regimes.
    \item
    While simple correlation-based linear estimators for quantitative group testing were previously studied in both the noiseless setting~\cite{gebhard2019quantitative}\cite{gebhard2022parallel} and noisy settings~\cite{hahn-klimroth_distributed_2022}, our results provide sharper guarantees.
    In particular, these prior formulations allow multi-set defectives within pooled tests, whereas we restrict attention to standard binary testing matrices, aligning more closely with the classical definition of QGT.
    Despite this stronger modeling restriction, our bounds achieve equal or improved sample complexity in all three models, and we derive explicit non-asymptotic guarantees with concrete absolute constants.
\end{itemize}
\textbf{Organization.} The remainder of the paper is organized as follows.
Section~\ref{sec:preliminaries} introduces the QGT problem formulation and the three observation models considered.
Section~\ref{sec:maintheorems} presents the main results of the paper.
Section~\ref{sec:proofs} provides the proofs of the main theorems.

\begin{table*}[t]
\vspace{0.2cm}
\centering
\scriptsize
\setlength{\tabcolsep}{4.5pt}
\renewcommand{\arraystretch}{1.35}

\begin{tabular}{p{2.2cm} p{3.6cm} p{4.6cm} p{7.0cm}}
\hline
\textbf{Setting} & \textbf{Type of Bound} & \textbf{Decoder / Method} & \textbf{Required number of tests $m$ (fractions should be rounded up)} \\
\hline

% =========================================================
\multirow{5}{*}{Noiseless}
& \multirow{4}{*}{Algorithmic Upper Bound}
& Linear estimator (Theorem~\ref{thm:noiselessQGT})
& $\displaystyle m = \Bigl(\frac{16}{\ln 3}k + 8 - \frac{8}{\ln 3}\Bigr)\log(k(n-k))$
\\

& & Gebhard~\cite{gebhard2019quantitative}
& $\displaystyle m = O\left(\frac{1+\sqrt{\theta}}{1-\sqrt{\theta}}k\log\frac{n}{k}\right)$
\\

& & Hahn-Klimroth, Muller~\cite{hahn-klimroth_near-optimal_2023,hahn-klimroth_near_2022}
& $\displaystyle m = O\left(\frac{1+\sqrt{\theta}}{1-\sqrt{\theta}}\frac{k\log(n/k)}{\log k}\right)$
\\

& & Soleymani and Javidi~\cite{soleymani_quantitative_2024}
& $\displaystyle m = (2e+\gamma)\frac{k}{\log k}\log\frac{n}{k}\log\log k$, $\gamma = 1-\exp(-1/2)$
\\

\cline{2-4}

& Information-theoretic Lower Bound
& Djackov et al.~\cite{djackov1975}
& $\displaystyle m = \Omega\left(2\,\frac{k\log(n/k)}{\log k}\right)$
\\

\hline

% =========================================================
\multirow{4}{*}{Additive Gaussian}
& \multirow{3}{*}{Algorithmic Upper Bound}
& LSE decoder (Theorem~\ref{thm:gaussianQGT})
& $\displaystyle m = (1+o(1))\frac{k\log(n/k)}{\log(1+k/\sigma^2)}$
\\

& & Linear estimator (Theorem~\ref{thm:gaussianQGT})
& $\displaystyle m = \Bigl(\frac{16}{\ln 3}k + 32\sigma^2 + 8 - \frac{8}{\ln 3}\Bigr)\log(k(n-k))$
\\

& & Hahn-Klimroth and Kasser~\cite{hahn-klimroth_distributed_2022}
& $\displaystyle m = O\left((1+\sqrt{\theta})^2 k \ln n\right)$
\\

\cline{2-4}

& Information-theoretic Lower Bound
& Theorem~\ref{thm:gaussianQGT}
& $\displaystyle m = \Omega\!\left(\frac{k\log(n/k)}{\log(1+k/4\sigma^2)}\right)$
\\

\hline

% =========================================================
\multirow{4}{*}{Noisy Z-channel}
& \multirow{3}{*}{Algorithmic Upper Bound}
& LSE decoder (Theorem~\ref{thm:ZchannelQGT})
& $\displaystyle m = \frac{k\log(n/k)}{\log(1+2C_p)}$, $C_p=\frac{(1-p)^2\ln\frac{p}{1-p}}{2p-1}$
\\

& & Linear estimator (Theorem~\ref{thm:ZchannelQGT})
& $\displaystyle m =
\Bigl(
\frac{16p}{(1-p)^2\ln\frac{1+p}{1-p}}
+
\frac{8(2k-1)(1+p)}{(1-p)^2\ln\frac{3+p}{1-p}}
\Bigr)\log(k(n-k))$
\\

& & Hahn-Klimroth and Kasser~\cite{hahn-klimroth_distributed_2022}
& $\displaystyle m =O\left(\frac{(1+\sqrt{\theta})^2}{1-p}k\ln n\right)$
\\

\cline{2-4}

& Information-theoretic Lower Bound
& Theorem~\ref{thm:ZchannelQGT}
& $\displaystyle m =\Omega\left( \frac{k\log(n/k)}{\log(\frac{1-k/n+kp/n}{p})}\right)$
\\
\hline
\\
\end{tabular}

\caption{Comparison of our results with existing work.}
\label{tab:result summary}
\end{table*}

\section{Preliminaries}\label{sec:preliminaries}
\subsection{Basic Notations}
For a matrix $A$, we use $a_{i,j}$ to denote its $(i, j)$-th entry. $A_i$ denotes the $i^{\text{th}}$ row of $A$. $|A_i|$ denotes the number of 1s in $A_i$. For a vector $v$, we use $v_i$ to denote its $i$-th entry. Other notations will be described in corresponding context.

\subsection{Quantitative Group Testing}
Let $n$ denote the \emph{total number of items}, and let $k$ denote the \emph{number of defective items}. Let $x^*$ be an $n$-dimensional binary vector with Hamming weight $k$. i.e. $\|x^*\|_0=k$. We denote the support of $x^*$ by
\[
D := \supp(x^*) \subset [n], \qquad |D| = k.
\]
We remark that, although our algorithms will only consider input vectors $x^*$ with Hamming weight exactly $k$ for convenience, these can be generalized to the case $\|x^*\|_0\leq k$ without order-wise change in algorithm complexities. 

Let $m$ be the \emph{number of tests}, and let $A$ be a $m\times n$ binary matrix which we call the \emph{test matrix}. 
Let $y$ be the \emph{test outcome vector}. This vector will be defind differently in different models, and will be specified in the next section. 

Unless stated otherwise, we assume a random Bernoulli design for test matrix:
\[
a_{i,j} \stackrel{\text{i.i.d.}}{\sim} \Ber(1/2).
\]

The goal of QGT is to recover the support $D^*$ (equivalently, the vector $x^*$) from the test outcomes $y\in\mathbb{R}^m$, using as few tests $m$ as possible, with vanishing error probability as $n\to\infty$. Denote $\widehat{x}$ as the outcome of the decoding algorithm. In this paper, for all cases, we aim for \emph{exact recovery}, i.e. $\widehat{x}=x^*$. 

% Let $n$ denote the total number of items.
% The unknown signal is a binary vector
% \[
% x^* \in \{0,1\}^n,
% \qquad
% \|x^*\|_0 = k,
% \]
% where $k = n^\theta$ for some fixed $\theta\in(0,1)$.
% We denote the support of $x^*$ by
% \[
% D^* := \supp(x^*) \subset [n], \qquad |D| = k.
% \]
% We perform $m$ pooled tests, specified by a testing matrix
% \[
% A \in \{0,1\}^{m\times n}.
% \]
% The $i$-th row $A_i \in \{0,1\}^n$ indicates which items are included in test $i$.
% Unless stated otherwise, we assume a random Bernoulli design:
% \[
% A_{i,j} \stackrel{\text{i.i.d.}}{\sim} \Ber(1/2).
% \]
% The goal of QGT is to recover the support $D^*$ (equivalently, the vector $x^*$) from the test outcomes $y\in\mathbb{R}^m$, using as few tests $m$ as possible, with vanishing error probability as $n\to\infty$.
\subsection{Observation Models}
\paragraph{Noiseless model}
In the noiseless setting, each test returns the exact number of defectives in the pool:
\begin{equation}
y = A x^*,
\qquad
y_i = \sum_{j\in D} a_{i,j}.
\label{eq:noiseless}
\end{equation}
\paragraph{Additive Gaussian noise model}
In this model, the pooled counts are corrupted by independent additive Gaussian noise:
\begin{equation}
y = A x^* + N,
\qquad
N_i \stackrel{\text{i.i.d.}}{\sim} \mathcal{N}(0,\sigma^2), \qquad 1\leq i\leq n.
\label{eq:gaussian}
\end{equation}

\paragraph{Noisy Z-channel model}
In the noisy Z-channel model, each defective item included in a test may independently fail to contribute to the observed count.
Formally, for each test $i$ and item $j$, let
\[
z_{i,j} \stackrel{\text{i.i.d.}}{\sim} \Ber(1-p),
\]
where $p\in(0,1)$ is the false-negative probability.
The observation is
\begin{equation}
y_i = \sum_{j=1}^n a_{i,j}\,x_j^*\,z_{i,j}.
\label{eq:zchannel}
\end{equation}
% Equivalently, conditional on $A$ and $x^*$,
% \[
% \E[y \mid A, x^*] = (1-p) A x^*.
% \]

\subsection{Decoding methods}
We analyze two decoding strategies that are common across all three models.

\subsubsection{Linear estimator (correlation-based decoder)}
The linear estimator assigns each item a score based on its correlation with the observed outcomes, which is widely used in literature~\cite{hahn-klimroth_distributed_2022,mazumdar2025exact,gebhard2019quantitative}.
Define the score vector
\begin{equation}
S := A^\top y \in \mathbb{R}^n,
\qquad
S_j = \sum_{i=1}^m a_{i,j} y_i.
\label{eq:score}
\end{equation}
The decoder outputs the indices of the $k$ largest coordinates of $S$:
\begin{equation}
\widehat{D}_{\mathrm{lin}} := \mathrm{TopK}(S,k),
\qquad
% \widehat{x}^{\mathrm{lin}}_j := \mathbf{1}\{j\in \widehat{D}_{\mathrm{lin}}\}.
\label{eq:linear-decoder}
\end{equation}
Intuitively, defective items tend to participate in tests with larger observed counts and therefore accumulate larger scores.
This decoder is computationally efficient and runs in polynomial time.

\subsubsection{Least Squares Estimator}
We also study a decoder motivated by least squares.
In all three models, the expected observation conditioned on $A$ and a candidate signal $x$ takes the form
\[
\E[y \mid A, x] =
\begin{cases}
A x, & \text{noiseless},\\
A x, & \text{Gaussian},\\
(1-p) A x, & \text{Z-channel}.
\end{cases}
\]
This motivates the least-squares objective
\[
\mathcal{L}(x) := \|y - \E[y \mid A, x]\|_2^2.
\]
The LSE decoder is defined as
\begin{equation}
\widehat{x}^{\mathrm{LSE}}
\in
\arg\min_{x\in\{0,1\}^n,\ \|x\|_0 = k}
\|y - \E[y \mid A, x]\|_2^2.
\label{eq:lse}
\end{equation}
Concretely:
\begin{itemize}
\item In the noiseless and Gaussian models, $\E[y \mid A, x] = A x$, leading to the objective $\|y-Ax\|_2^2$.
\item In the noisy Z-channel model, $\E[y \mid A, x] = (1-p)Ax$, yielding the objective $\|y-(1-p)Ax\|_2^2$.
\end{itemize}
While solving \eqref{eq:lse} exactly is generally combinatorial and not computationally efficient, the LSE serves as a statistically optimal benchmark.

% \subsection{Performance criterion}
% A decoder $\widehat{x}$ is said to achieve \emph{exact recovery} (w.h.p) if
% \[
% \Pr(\widehat{x} = x^*) \to 1
% \qquad \text{as } n\to\infty.
% \]
% Our goal is to characterize the minimum number of tests $m$ required to achieve exact recovery under each model and each decoder.

\section{Main Results}\label{sec:maintheorems}
\subsection{Noiseless quantitative group testing}
Our first theorem establishes that a simple linear estimator based on correlation scores suffices for exact recovery with high probability.

\begin{theorem}[Noiseless QGT: Linear Estimator]\label{thm:noiselessQGT}
Let $n,k,A,x^*,y$ be as defined in Section~\ref{sec:preliminaries}. 
In the noiseless model $y = Ax^*$, consider the linear decoder that outputs indices of the $k$ largest coordinates of $A^\top y$,
the number of tests required to recover $x^*$ with vanishing error probability satisfies
\[
m\ge \bigl(\frac{16}{\ln3}k+8-\frac{8}{\ln3}\bigr)\log\bigl(k(n-k)\bigr).
\]
\end{theorem}

\subsection{Additive Gaussian noise model}
In this setting, we show that least-squares maximum-likelihood decoding achieves the optimal sample complexity, and we complement this with a matching information-theoretic converse.
\begin{theorem}[Additive Gaussian QGT]\label{thm:gaussianQGT}
Let $n,k,A,x^*,y$ be as defined in Section~\ref{sec:preliminaries} and let $k=n^\theta$ for some $\theta\in(0,1)$. In the additive Gaussian noise model $y = Ax^* + N$, where $N$ has i.i.d.\ Gaussian entries with variance $\sigma^2$, we have:
\begin{itemize}

  \item For the LSE decoder
    \[
    \widehat{x} = \arg\min_x \|y - A x\|_2^2,
    \] 
    the number of tests required to recover $x^*$ with vanishing error probability satisfies
  \[
  m = O\!\left(\frac{k\log(n/k)}{\log(1 + k/\sigma^2)}\right).
  \]
  \item Any decoder with any test design must satisfy 
  \[m = \Omega\!\left(\frac{k\log(n/k)}{\log(1 + k/\sigma^2)}\right).\]
  \item For the linear decoder,
  the number of tests required to recover $x^*$ with vanishing error probability satisfies
    \[
    m\ge \bigl(\frac{16}{\ln3}k+32\sigma^2+8-\frac{8}{\ln3}\bigr)\log\bigl(k(n-k)\bigr).
    \]
\end{itemize}
\end{theorem}

\subsection{Noisy Z-channel model}
In this asymmetric noise setting, we provide achievability guarantees and information-theoretic limits.
\begin{theorem}[Noisy Z-channel QGT]\label{thm:ZchannelQGT}
Let $n,k,A,x^*,y$ be as defined in Section~\ref{sec:preliminaries}. In the noisy Z-channel model with false-negative probability $p$, we have: 
\begin{itemize}
  \item For the LSE decoder
    \[
    \widehat{x} = \arg\min_x \|y - (1-p)A x\|_2^2,
    \] 
    the number of tests required to recover $x^*$ with vanishing error probability satisfies
  \[
  m = O\!\left(\frac{k\log(n/k)}{\log(1 + 2C_p)}\right)
  \]
  tests, where $C_p$ depends explicitly on the channel parameter $p$.
  $C_p=\frac{(1-p)^2\ln\frac{p}{1-p}}{2p-1}$
  \item If $\E \!\left[\frac{1}{\max\{A_ix, 1\}} \right]\leq \frac{c}{k}$, for some absolute constant $c$, then any decoder with any test design must satisfy
  \[
  m = \Omega\!\left(\frac{k\log(n/k)}{\log(\frac{1-k/n+kp/n}{p})}\right),
  \]
  The assumption is satisfied if every entry of test matrix $A$ is i.i.d $\Ber(1/2)$.
  \item For the linear decoder,
  the number of tests required to recover $x^*$ with vanishing error probability satisfies
    \[
    m\ge \Bigl(
    \frac{16p}{(1-p)^2 \ln\frac{1+p}{1-p}}
    +
    \frac{8(2k-1)(1+p)}{(1-p)^2 \ln\frac{3+p}{1-p}}\Bigr) \log\bigl(k(n-k)\bigr).
    \]

\end{itemize}
\end{theorem}
% =======================

\section{Proofs of Theorem \ref{thm:noiselessQGT}, \ref{thm:gaussianQGT} and \ref{thm:ZchannelQGT}}\label{sec:proofs}
\subsection{Proof of Theorem \ref{thm:noiselessQGT}}
\begin{proof}
Write $D := \supp(x^*)$, $|D|=k$, and for each $i\in[n]$ define the score
\[
S_i := (A^\top y)_i
= \sum_{t=1}^m a_{t,i}\,y_t,
\qquad
y_t = \sum_{s\in D} a_{t,s}.
\]
For $i\in D$ and a fixed row $t$,
\[
a_{t,i} y_t
= a_{t,i}\sum_{s\in D} a_{t,s}
=
\begin{cases}
0, & a_{t,i}=0,\\[3pt]
1 + \sum_{s\in D\setminus\{i\}} a_{t,s}, & a_{t,i}=1,
\end{cases}
\]
where the $a_{t,s}$ for $s\in D\setminus\{i\}$ are i.i.d.\ $\Ber(1/2)$ and
independent of $a_{t,i}$. Hence
\[
\mathbb{E}[a_{t,i}y_t]
= \frac12\Bigl(1 + \frac{k-1}{2}\Bigr)
= \frac{k+1}{4},
\qquad
\mathbb{E}[S_i] = m\,\frac{k+1}{4},
\]
For $j\notin D$, $a_{t,j}$ is independent of $y_t$, and
$\mathbb{E}[y_t] = k/2$, so
\[
\mathbb{E}[a_{t,j}y_t]
= \mathbb{E}[a_{t,j}]\,\mathbb{E}[y_t]
= \frac12\cdot\frac{k}{2}
= \frac{k}{4},
\qquad
\mathbb{E}[S_j] = m\,\frac{k}{4},
\]
Thus, for every $i\in D$ and $j\notin D$,
\[
\mathbb{E}[S_i - S_j] = \frac{m}{4}.
\]
\\\\
Fix $i\in D$ and $j\notin D$ and write
\[
S_i - S_j = \sum_{t=1}^m z_t,
\qquad
z_t := a_{t,i}\sum_{s\in D} a_{t,s}
      - a_{t,j}\sum_{s\in D} a_{t,s}.
\]
Each summand admits the decomposition
\[
z_t 
= a_{t,i}^2 
  + \sum_{s\in D\setminus\{i\}} a_{t,i}a_{t,s}
  - \sum_{s\in D} a_{t,j}a_{t,s},
\]
where $a_{t,i}^2\sim\Ber(1/2)$ and every product $a_{t,i}a_{t,s}$ or
$a_{t,j}a_{t,s}$ is $\Ber(1/4)$.
Using the optimal subgaussian variance proxy for centered Bernoulli
variables \text{\cite{kearns_large_2013}},
\[
\sigma^2_{\mathrm{opt}}(p)
=\frac{1-2p}{2\ln\frac{1-p}{p}}
\]
and summing over the independent Bernoulli components of $z_t$, we obtain the variance proxy
\[
\|z_t - \mathbb{E}z_t\|^2_\text{vp}
\;\le\;
\frac14 + \frac{2k-1}{4\ln 3},
\]
uniformly in $t$.
\\\\
Denote $\mu := \mathbb{E}[S_i - S_j] = m/4$.
The standard subgaussian tail bound gives
\[
\Pr(S_i < S_j)
= \Pr(S_i - S_j \le 0)
= \Pr\Bigl(\sum_{t=1}^m (z_t - \mathbb{E}z_t) \le -\mu\Bigr)
\le 
\exp\Bigl(
  - \frac{\mu^2}{2 m \|z_t - \mathbb{E}z_t\|^2_\text{vp}}
\Bigr)
\le
\exp\Bigl(
  - \frac{m}{32 (\frac14 + \frac{2k-1}{4\ln3})}
\Bigr).
\]
An error occurs if $S_i \le S_j$ for some $i\in D$, $j\notin D$.  
By the union bound,
\[
\Pr(\text{Error})
\;\le\;
\sum_{i\in D}\sum_{j\notin D} \Pr(S_i < S_j)
\;\le\;
k(n-k)
\exp\Bigl(
  - \frac{m}{32 (\frac14 + \frac{2k-1}{4\ln3})}
\Bigr)
\]
Therefore, to recover $x^*$ with vanishing error probability, we need 
\[
m\ge \bigl(\frac{16}{\ln3}k+8-\frac{8}{\ln3}\bigr)\log\bigl(k(n-k)\bigr).
\]
\end{proof}

\subsection{Proof of Theorem \ref{thm:gaussianQGT}}
\subsubsection{LSE achievability}
\begin{proof}
Minimizing over $x$ is equivalent to maximizing
\[
f(x) := \sum_{i=1}^m [2 y_i A_i x - (A_i x)^2],
\]
so that the decoder is $\widehat{x} = \arg\max_x f(x)$.  
\\\\
For any $x \ne x^*$, define the difference
\[
f(x) - f(x^*) = \sum_{i=1}^m \Delta_i,
\quad
\text{where } \Delta_i = 2(A_i x - A_i x^*)(y_i - A_i x^*) - (A_i x - A_i x^*)^2.
\]
Substituting $y_i = A_i x^* + N_i$ gives
\[
\Delta_i = 2(A_i x - A_i x^*)N_i - (A_i x - A_i x^*)^2.
\]
Let $b := A x - A x^*$.  Then each $\Delta_i = 2 b_i N_i - b_i^2$, and the error event is
\[
\Pr(\text{error})
= \Pr\!\left( \exists x \ne x^* : f(x) \ge f(x^*) \right)
= \sum_{l=1}^k \sum_{\|x-x^*\|_H = 2l} \Pr(f(x)\ge f(x^*)).
\]
Conditioning on the values of $b' = (b_1',\ldots,b_m')$, we have
\begin{align*}
\Pr(f(x)\ge f(x^*))
&= \Pr\!\left(\sum_{i=1}^m 2b_i N_i \ge \sum_{i=1}^m b_i^2 \right)\\
&= \sum_{b'} \Pr\!\left(\sum_{i=1}^m 2b_i' N_i \ge \sum_{i=1}^m b_i'^2\right)
    \Pr(b' = A(x - x^*)).
\end{align*}
Since each $b_i N_i$ is a mean-zero sub-Gaussian variable with parameter $\|b_i N_i\|_{\psi_2}^2 = \frac{8}{3} b_i^2 \sigma^2$, and the terms are independent, by the standard Hoeffding bound (Theorem~2.6.2 in~\cite{vershynin2009high}), for some absolute constant $c>0$ we have
\[
\Pr\!\left(\sum_{i=1}^m 2b_i N_i \ge \sum_{i=1}^m b_i^2 \right)
\le 2\exp\!\left(
   -c \frac{(\sum_i b_i^2)^2}{\sum_i b_i^2\sigma^2}
\right)
= 2\exp\!\left(-c \frac{\sum_i b_i^2}{\sigma^2}\right).
\]
\\\\
Each test row $A_i$ is independent and has i.i.d.\ entries $\mathrm{Ber}(1/2)$.  
For fixed $x$ and $x^*$ differing in $2l$ positions, the random variable $b_i = A_i(x-x^*)$ is the difference of two independent $\mathrm{Bin}(l,1/2)$ variables.  
By Lemma \ref{lemma_exp},
\[
\Pr(b_i = t) = 4^{-l} {2l \choose l + |t|} \le \frac{e}{\pi\sqrt{2l}}\exp\!\left(-\frac{t^2}{2l}\right).
\]
As the rows $A_i$ are independent,
\[
\Pr(b' = A(x-x^*)) = \prod_{i=1}^m \Pr(b_i'=b_i)
\le \left(\frac{e}{\pi\sqrt{2l}}\right)^m
    \exp\!\left(-\frac{\sum_i b_i'^2}{2l}\right).
\]
Combining the two bounds,
\begin{align*}
\Pr(f(x)\ge f(x^*))
&\le \sum_{b'} 2\exp\!\left(-c\frac{\sum_i b_i'^2}{\sigma^2}\right)
        \left(\frac{e}{\pi\sqrt{2l}}\right)^m
        \exp\!\left(-\frac{\sum_i b_i'^2}{2l}\right)\\
&= 2\left(\frac{e}{\pi\sqrt{2l}}\right)^m
    \sum_{b'} \exp\!\left(-(\frac{c}{\sigma^2}+\frac{1}{2l})\sum_i b_i'^2\right).
\end{align*}
Since each coordinate $b_i'\in[-l,l]$ and they are independent,
\[
\sum_{b'} \exp\!\left(-(\frac{c}{\sigma^2}+\frac{1}{2l})\sum_i b_i'^2\right)
= \prod_{i=1}^m \sum_{t=-l}^l \exp\!\left(-(\frac{c}{\sigma^2}+\frac{1}{2l})t^2\right).
\]
Bounding the discrete sum by an integral,
\[
\sum_{t=-l}^l e^{-\alpha t^2}
\le 2\int_{0}^\infty e^{-\alpha t^2}dt
= \frac{\sqrt{\pi}}{2\sqrt{\alpha}},
\quad \alpha = \frac{c}{\sigma^2}+\frac{1}{2l}.
\]
Hence,
\[
\Pr(f(x)\ge f(x^*))
\le \left(\frac{e}{\pi\sqrt{2l}}\right)^m
     \left(\frac{\sqrt{\pi}}{2\sqrt{\frac{c}{\sigma^2}+\frac{1}{2l}}}\right)^m
= \left(\frac{c'}{\sqrt{\frac{l}{\sigma^2}+1}}\right)^m,
\]
for some absolute constant $c'>0$.
For a given Hamming distance $2l$, the number of possible $x$ differing from $x^*$ in those coordinates is at most ${k \choose l}{n-k \choose l} \le {k \choose l}{n \choose l}$.  
Applying Lemma \ref{lemma_N_choose_M} gives
\[
{k \choose l}{n \choose l} \le
\left(\frac{e^2kn}{l^2}\right)^l.
\]
% \textcolor{blue}{There is an easier way to analyze this, by directly analyzing the term ${k \choose l}{n-k \choose l}$ instead of the term $\left(\frac{e^2kn}{l^2}\right)^l$. 
% \\
% It can be shown that, when $n>k^2+k$, $f(l):={k \choose l}{n-k \choose l}$ achieves maximum at $l=k$. 
% \\
% In general, $f(l)$ must achieve its maximum for some $k/2 \leq l\leq k$, regardless of the regime. Plugging in such $l$ will get an desired result always. }
% \\\\
Thus,
\begin{align*}
\Pr(\text{error})
&\le k \max_{1\le l\le k}
     {k \choose l}{n \choose l}
     \left(\frac{c'}{\sqrt{\frac{l}{\sigma^2}+1}}\right)^m\\
&\le k \max_{1\le l\le k}
     \left(\frac{e^2kn}{l^2}\right)^l
     \left(\frac{c'}{\sqrt{\frac{l}{\sigma^2}+1}}\right)^m.
\end{align*}
\\\\
To make the total probability vanish as $n\to\infty$, it suffices that the exponent of the largest term be negative, yielding
\[
m \ge \max_{1\le l\le k}
     \frac{l\log(\frac{kn}{l^2})}{\log(\frac{l}{\sigma^2}+1)}.
\]
Let
\[
T\ :=\ \frac{k\log(n/k)}{\log\!\big(1+\tfrac{k}{\sigma^2}\big)}
\ =\ \frac{k(1-\theta)\log n}{\log\!\big(1+\tfrac{k}{\sigma^2}\big)}.
\]
Now we aim to show that for any integer $l\in[1,k]$,
\[
F(l):=\frac{l\log(\frac{kn}{l^2})}{\log(\frac{l}{\sigma^2}+1)} \le 
(1+o(1))\ T.
\]
\smallskip
\emph{Case 1: $1\le l\le k^{1-\varepsilon}$.}
Using $k=n^\theta$ and $l\le k^{1-\varepsilon}$, we have
\[
\log\!\Big(\frac{kn}{l^2}\Big)\ \le\ \log(kn)= (1+\theta)\log n,
\]
and $\log(1+\tfrac{l}{\sigma^2})\ge \log(1+\tfrac{1}{\sigma^2})$. Hence
\[
F(l)\ \le\ \frac{k^{1-\varepsilon}(1+\theta)\log n}{\log\!\big(1+\tfrac{1}{\sigma^2}\big)}.
\]
Then
\[
\frac{F(l)}{T}
\ \le\
\frac{1+\theta}{1-\theta}\cdot \frac{1}{k^{\varepsilon}}\cdot
\frac{\log\!\big(1+\tfrac{k}{\sigma^2}\big)}{\log\!\big(1+\tfrac{1}{\sigma^2}\big)}.
\]
For any fixed $\varepsilon>0$, the factor $k^{-\varepsilon}$ dominates the at-most polylogarithmic growth of $\log(1+\tfrac{k}{\sigma^2})$, so $\frac{F(l)}{T}\to 0$. Therefore, for all sufficiently large $n$,
\[
\max_{1\le l\le k^{1-\varepsilon}} F(l)\ \le\ o(1)\cdot T
\]
\smallskip
\emph{Case 2: $k^{1-\varepsilon}< l\le k$.}
Set $g(l):=l\log\!\big(\tfrac{kn}{l^2}\big)$. Then
\[
g'(l)\ =\ \log(kn)-2\log l-2,\qquad
g''(l)\ =\ -\frac{2}{l}<0,
\]
so $g$ is concave. When $n/k>e^2$, we have $g'(k)=\log(n/k)-2>0$, and concavity implies $g$ attains its maximum over $[1,k]$ at the boundary $l=k$. Hence, for all sufficiently large $n$,
\[
l\log\!\Big(\frac{kn}{l^2}\Big)\ \le\ k\log\!\Big(\frac{n}{k}\Big).
\]
For the denominator, we have
\[
\log\!\Big(1+\frac{l}{\sigma^2}\Big)\ \ge\ \log\!\Big(1+\frac{k^{1-\varepsilon}}{\sigma^2}\Big)
\ =\ (1-\varepsilon)\log\!\Big(1+\frac{k}{\sigma^2}\Big)\ +\ O(1),
\]
Therefore,
\[
F(l)\ \le\ \frac{k\log(n/k)}{(1-\varepsilon)\log\!\big(1+\tfrac{k}{\sigma^2}\big)+O(1)}
\ =\ (1+o(1))\,\frac{k\log(n/k)}{\log\!\big(1+\tfrac{k}{\sigma^2}\big)}.
\]
\end{proof}
\subsubsection{Converse Bound}
\begin{proof}
Let $X$ denote the random support (uniform on $\binom{n}{k}$), $Y=(y_1,\dots,y_m)$ the observations, and condition on any fixed matrix $A$.  
By Fano’s inequality,
\[
P_e \;\ge\; 1 - \frac{I(X;Y\,|\,A)+\log 2}{\log \binom{n}{k}},
\]
so that achieving $P_e\to 0$ requires
\[
I(X;Y\,|\,A)\ \ge\ (1-o(1))\,\log\binom{n}{k}.
\]
For each row $A_i$, define the noiseless measurement
\[
S_i := \langle A_i, X\rangle,
\]
then $y_i = S_i + N_i$ with $N_i\sim\mathcal{N}(0,\sigma^2)$ independent of $S_i$.  
Because the channel is memoryless conditional on $A$, the chain rule gives
\[
I(X;Y\,|\,A)
= \sum_{i=1}^m I(X;y_i\,|\,A,y_{<i})
\le \sum_{i=1}^m I(S_i;y_i\,|\,A_i),
\]
For each test, we have
\[
I(S_i;y_i\,|\,A_i)
= h(y_i\,|\,A_i) - h(N_i).
\]
\[\mathrm{Var}(y_i\,|\,A_i) = \mathrm{Var}(S_i\,|\,A_i) + \sigma^2\]
For any random variable $Z$ with variance $v$, the differential entropy satisfies
\[
h(Z)\ \le\ \tfrac{1}{2}\log(2\pi e v),
\]
with equality if and only if $Z$ is Gaussian.  
Hence, we obtain
\[
h(y_i\,|\,A_i) \;\le\; \tfrac{1}{2}\log \bigl(2\pi e(\mathrm{Var}(S_i\,|\,A_i)+\sigma^2)\bigr),
\qquad
h(N_i) = \tfrac{1}{2}\log(2\pi e\sigma^2).
\]
Therefore,
\[
I(S_i;y_i\,|\,A_i)
\;\le\;
\tfrac{1}{2}\log\Bigl(1+\frac{\mathrm{Var}(S_i\,|\,A_i)}{\sigma^2}\Bigr).
\]
Given $A_i$ and the uniform random support $X$, $S_i = \langle A_i,X\rangle$ follows a hypergeometric distribution with parameters $(n,k,w_i)$, where $w_i=\|A_i\|_0$. Its variance satisfies
\[
\mathrm{Var}(S_i\,|\,A_i)
= k\frac{w_i}{n}\Bigl(1-\frac{w_i}{n}\Bigr)\frac{n-k}{n-1}
\le \frac{k}{4},
\]
the upper bound being attained when $w_i=n/2$.  
\\\\
Consequently,
\[
I(S_i;y_i\,|\,A_i)
\le \tfrac{1}{2}\log\!\Bigl(1+\tfrac{k}{4\sigma^2}\Bigr).
\]
Summing over all tests gives
\[
I(X;Y\,|\,A)
\le m\cdot \tfrac{1}{2}\log\!\Bigl(1+\tfrac{k}{4\sigma^2}\Bigr).
\]
Combining this with the Fano requirement $I(X;Y\,|\,A)\ge (1-o(1))\log\binom{n}{k}$ yields
\[
m \;\ge\; (1-o(1))\,\frac{\log\binom{n}{k}}{\tfrac{1}{2}\log(1+\tfrac{k}{4\sigma^2})}.
\]
Finally, using $\log\binom{n}{k}=k\log(n/k)+O(k)$ gives the desired bound.
\end{proof}
\subsubsection{Linear estimator achievability}
\begin{proof}
This proof is essentially the same as the proof of Theorem~\ref{thm:noiselessQGT}, for completeness, we present the proof in the following: 

Write $D := \supp(x^*)$, $|D|=k$, and for each $i\in[n]$ define the score
\[
S_i := (A^\top y)_i
= \sum_{t=1}^m a_{t,i}\,y_t,
\qquad
y_t = \sum_{s\in D} a_{t,s} + N_t.
\]
For $i\in D$,
\[
\mathbb{E}[a_{t,i}y_t]
= \frac12\Bigl(1 + \frac{k-1}{2} + 0\Bigr)
= \frac{k+1}{4},
\qquad
\mathbb{E}[S_i] = m\,\frac{k+1}{4}.
\]
For $j\notin D$,
\[
\mathbb{E}[a_{t,j}y_t]
= \mathbb{E}[a_{t,j}]\,\mathbb{E}\Bigl[\sum_{s\in D} a_{t,s} + w_t\Bigr]
= \frac12\cdot\frac{k}{2}
= \frac{k}{4},
\qquad
\mathbb{E}[S_j] = m\,\frac{k}{4}.
\]
Hence, for every $i\in D$ and $j\notin D$,
\[
\mathbb{E}[S_i - S_j] = \frac{m}{4}.
\]
Fix $i\in D$ and $j\notin D$ and write
\[
S_i - S_j = \sum_{t=1}^m z_t,
\qquad
z_t := a_{t,i}y_t - a_{t,j}y_t.
\]
Using $y_t = \sum_{s\in D} a_{t,s} + N_t$, we decompose
\[
z_t 
= \underbrace{\Bigl(a_{t,i}\sum_{s\in D} a_{t,s}
  - a_{t,j}\sum_{s\in D} a_{t,s}\Bigr)}_{\text{Noiseless part}}
  \;+\;
  \underbrace{(a_{t,i}-a_{t,j})N_t}_{\text{Gaussian noise part}}.
\]
By independence of the discrete and Gaussian components and additivity of variance proxies for independent sums, we have the uniform bound
\[
\|z_t - \mathbb{E}z_t\|^2_\text{vp}
\;\le\;
\frac14 + \frac{2k-1}{4\ln 3} + \sigma^2.
\]
Let $\mu := \mathbb{E}[S_i - S_j] = m/4$.  
Since the $z_t$ are independent and subgaussian with the above proxy, the standard subgaussian tail bound (Lemma 1.3 in \cite{rigollet_chapter_nodate}) gives
\[
\Pr(S_i < S_j)
= \Pr\Bigl(\sum_{t=1}^m (z_t - \mathbb{E}z_t) \le -\mu\Bigr)
\le 
\exp\Bigl(
  - \frac{\mu^2}{2 \sum_{t=1}^m \|z_t - \mathbb{E}z_t\|^2_\text{vp}}
\Bigr)
\le
\exp\Bigl(
  - \frac{m}{32(\frac14 + \frac{2k-1}{4\ln 3} + \sigma^2)}
\Bigr),
\]
An error occurs if $S_i \le S_j$ for some $i\in D$, $j\notin D$.  
By the union bound,
\[
\Pr(\text{Error})
\;\le\;
\sum_{i\in D}\sum_{j\notin D} \Pr(S_i < S_j)
\;\le\;
k(n-k)
\exp\Bigl(
  - \frac{m}{32(\frac14 + \frac{2k-1}{4\ln 3} + \sigma^2)}\Bigr).
\]
Therefore, to recover $x^*$ with vanishing error probability, we need 
\[
m\ge \bigl(\frac{16}{\ln3}k+32\sigma^2+8-\frac{8}{\ln3}\bigr)\log\bigl(k(n-k)\bigr).
\]
\end{proof}

\subsection{Proof of Theorem \ref{thm:ZchannelQGT}}
\subsubsection{LSE achievability}
\begin{proof}
Minimizing $\|y-(1-p)Ax\|_2^2$ is equivalent to maximizing
\[
f(x)=\sum_{i=1}^m \bigl[2y_i A_i x-(1-p)(A_i x)^2\bigr],
\]
so $\widehat{x}=\arg\max_x f(x)$.
For $x\neq x^*$ define
\[
f(x)-f(x^*)=\sum_{i=1}^m \Delta_i,
\qquad 
\Delta_i=2b_i\,(y_i-\E y_i)-(1-p)b_i^2,
\]
where $b_i:=A_i(x-x^*)$.  Since $y_i=\sum_j a_{i,j}x_j^*z_{i,j}$ with $z_{i,j}\sim\Ber(1-p)$,
each $y_i-\E y_i$ is a sum of at most $k$ centered $\Ber(1-p)$ variables.  
Thus its variance proxy satisfies
\[
\|y_i-\E y_i\|_{\mathrm{vp}}
\;\le\;
k\,\frac{2p-1}{2\ln\ \frac{p}{1-p}}.
\]
Applying Corollary~1.7 of \cite{rigollet_chapter_nodate} yields
\[
\Pr\!\left(\sum_{i=1}^m 2b_i(y_i-\E y_i)\ge (1-p)\sum_{i=1}^m b_i^2\right)
\le
\exp\!\left(-\,\frac{(1-p)^2\sum_i b_i^2}{k\cdot \frac{2p-1}{\ln\frac{p}{1-p}}}\right).
\]
Each test row $A_i$ is independent and has i.i.d.\ entries $\mathrm{Ber}(1/2)$.  
For fixed $x$ and $x^*$ differing in $2l$ positions, the random variable $b_i = A_i(x-x^*)$ is the difference of two independent $\mathrm{Bin}(l,1/2)$ variables.  
By Lemma \ref{lemma_exp},
\[
\Pr(b_i = t) = 4^{-l} {2l \choose l + |t|} \le \frac{e}{\pi\sqrt{2l}}\exp\!\left(-\frac{t^2}{2l}\right).
\]
As the rows $A_i$ are independent,
\[
\Pr(b = A(x-x^*)) = \prod_{i=1}^m \Pr(b_i)
\le \left(\frac{e}{\pi\sqrt{2l}}\right)^m
    \exp\!\left(-\frac{\sum_i b_i^2}{2l}\right).
\]
Combining the two bounds,
\begin{align*}
\Pr(f(x)\ge f(x^*))
&\le \sum_b \exp(-\frac{(1-p)^2\sum_{i=1}^m b_i^2}{k\cdot \frac{2p-1}{\ln\frac{p}{1-p}}})
        \left(\frac{e}{\pi\sqrt{2l}}\right)^m
        \exp\!\left(-\frac{\sum_i b_i^2}{2l}\right)\\
&= 2\left(\frac{e}{\pi\sqrt{2l}}\right)^m
    \sum_b \exp\!\left(-\alpha\sum_i b_i^2\right).
\end{align*}
where $\alpha = \frac{(1-p)^2\ln\frac{p}{1-p}}{2p-1}\cdot \frac{1}{k}+\frac{1}{2l}$.
\\
Since each coordinate $b_i\in[-l,l]$ and they are independent,
\[
\sum_b \exp\!\left(-\alpha\sum_i b_i^2\right)
= \prod_{i=1}^m \sum_{t=-l}^l \exp\!\left(-\alpha t^2\right).
\]
Bounding the discrete sum by an integral,
\[
\sum_{t=-l}^l e^{-\alpha t^2}
\le 2\int_{0}^\infty e^{-\alpha t^2}dt
= \frac{\sqrt{\pi}}{2\sqrt{\alpha}},
\]
Hence,
\[
\Pr(f(x)\ge f(x^*))
\le \left(\frac{e}{\pi\sqrt{2l}}\right)^m
     \left(\frac{\sqrt{\pi}}{2\sqrt{\alpha}}\right)^m
= \left(\frac{e/\sqrt{\pi}}{\sqrt{\frac{(1-p)^2\ln\frac{p}{1-p}}{2p-1}\frac{2l}{k}+1}}\right)^m.
\]
For a given Hamming distance $2l$, the number of possible $x$ differing from $x^*$ in those coordinates is at most ${k \choose l}{n-k \choose l} \le {k \choose l}{n \choose l}$.  
Applying Lemma \ref{lemma_N_choose_M} gives
\[
{k \choose l}{n \choose l} \le
\left(\frac{e^2kn}{l^2}\right)^l.
\]
Thus, let $C_p=\frac{(1-p)^2\ln\frac{p}{1-p}}{2p-1}$, 
\begin{align*}
\Pr(\text{error})
&\le k \max_{1\le l\le k}
     {k \choose l}{n \choose l}
     \left(\frac{e/\sqrt{\pi}}{\sqrt{C_p\frac{2l}{k}+1}}\right)^m\\
&\le k \max_{1\le l\le k}
     \left(\frac{e^2kn}{l^2}\right)^l
     \left(\frac{e/\sqrt{\pi}}{\sqrt{C_p\frac{2l}{k}+1}}\right)^m.
\end{align*}
\\\\
To make the total probability vanish as $n\to\infty$, it suffices that the exponent of the largest term be negative, yielding
\[
m \ge \max_{1\le l\le k}
     \frac{l\log(\frac{kn}{l^2})}{\log(C_p\frac{2l}{k}+1)} = \frac{k\log(\frac{n}{k})}{\log(2C_p+1)}
\]
\end{proof}
\subsubsection{Converse Bound}
\begin{proof}
Suppose $A$ is some arbitrary random binary matrix and $x$ is chosen uniformly at random from the set of all $k$ sparse binary vectors. For each test $i$, let $S_i := A_i x$ denote the (noiseless) number of defectives included in that test and $|A_i|$ denotes the number of 1s in $A_i$ where $A_i$ denotes the $i^{\text{th}}$ row of $A$.  
Since $A_i$ is independent of $x$, 
\[
I(Y_i,A_i;X)=I(Y_i;X\mid A_i)=H(Y_i\mid A_i)-H(Y_i\mid X,A_i)
\]
We will first argue that the distribution of $S_i|A\sim \text{hypergoemetric}(n, |A_i|,k)$, i.e., the distribution of $S_i$ conditioned on $A$ is hypergeometric with parameters $(n,|A_i|,k)$. This is because 
\[
\Pr(S_i = s|A=a) = \frac{\binom{|a_i|}{s}\binom{n-|a_i|}{k-s}}{\binom{n}{k}}.
\]
Further, we have $Y_i|(S_i,A)\sim \mathrm{Bin}(S_i,(1-p))$. 
Suppose, $A$ is fixed as $a$. Then, by law of total variance
\begin{align*}
\text{Var}[Y_i] &= \mathbb{E}[\text{Var}[Y_i|S_i]]+ \text{Var}[\mathbb{E}[Y_i|S_i]]\\
& = \mathbb{E}[S_i p(1-p)] + \text{Var}[S_i(1-p)]\\
& = p(1-p)\mathbb{E}[S_i] + (1-p)^2\text{Var}[S_i]\\
& = p(1-p)\mathbb{E}[S_i] + (1-p)^2\text{Var}[S_i]\\
& =p(1-p)k\frac{|a_i|}{n} + (1-p)^2k\frac{|a_i|}{n}\frac{(n-|a_i)}{n}\frac{(n-k)}{n-1}\\
& \leq (1-p)k\frac{|a_i|}{n}\left(p+(1-p)\frac{(n-k)}{n}\right)\\
& = (1-p)k\frac{|a_i|}{n}\left(1-(1-p)\frac{k}{n}\right)\\
& \leq (1-p)k\frac{n}{n}\left(1-(1-p)\frac{k}{n}\right)\\
& =(1-p)k\left(1-(1-p)\frac{k}{n}\right).
\end{align*}

% We have $Y_i\sim \mathrm{Bin}(|A_i|,(1-p)\frac{k}{n})$, where $|A_i|$ represents the number of 1 in $A_i$.
A standard discrete entropy upper bound \cite{lmassey_entropy_1988} gives
\begin{align*}
H(Y_i|A_i)&\leq \frac12\E\left[\log\bigl(2\pi e (\text{Var}(Y_i|A_i)+\frac{1}{12})\bigr)\right]\\
&\leq \frac12\log\bigl(2\pi e (k(1-p)(1-\frac{k}{n}(1-p))+\frac{1}{12})\bigr)\\
&\stackrel{(a)}{\leq}\frac12\log\bigl(4\pi e (k(1-p)(1-\frac{k}{n}(1-p)))\bigr).
\end{align*} where $(a)$ holds for large $k$ and fixed $p$.
Next, by Theorem 2 in \cite{jacquet_entropy_1999}
\begin{align*}
-H(Y_i\mid X,A_i)
&=-H(Y_i\mid S_i)\\
&\le\E_{S_i}\!\left[\min\left\{-\frac12\log\!\bigl(2\pi e\,S_i p(1-p)\bigr),0\right\}\right]\\
&\le -\frac12\log\!\bigl(2\pi e p(1-p)\bigr) +\frac12\E_{S_i}\!\left[\min\left\{\log\!\left(\frac{1}{S_i} \right), 0\right\}\right]\\
&\stackrel{(a)}{=} -\frac12\log\!\bigl(2\pi e p(1-p)\bigr) +\frac12\E_{S_i}\!\left[\log\!\left(\frac{1}{\max\{S_i, 1\}} \right)\right]\\
&\leq -\frac12\log\!\bigl(2\pi e p(1-p)\bigr) +\frac12\log\!\left(\E_{S_i}\!\left[\frac{1}{\max\{S_i, 1\}} \right]\right)\\
&\stackrel{(b)}{\leq} -\frac12\log\!\bigl(2\pi e p(1-p)\bigr) +\frac12\log\!\left(\frac{C}{k}\right)\\
&= -\frac12\log\!\left(\frac{2\pi e}{C} p(1-p)k\right). 
\end{align*} where $(a)$ holds because $S_i$ takes non negative integer values. $(b)$ uses the assumption $\E_{S_i}\!\left[\frac{1}{\max\{S_i, 1\}} \right]\leq \frac{C}{k}$ for some constant $C$.
% That is,
% \begin{align*}
% -H(Y_i\mid X,A_i)
% &\le\max\left\{-\E_{S_i}\!\left[\frac12\log\!\bigl(2\pi e p(1-p)\bigr)\right] + \frac12\E_{S_i}\left[\frac{1}{S_i}\right],0\right\}\\
% & \le \max\left\{-\E_{S_i}\!\left[\frac12\log\!\bigl(2\pi e p(1-p)\bigr)\right], 0\right\} + \max\left\{\frac12\E_{S_i}\left[\frac{1}{S_i}\right],0\right\}\\
% \end{align*}
Combining the two bounds,
\begin{align*}
I(Y_i,A_i;X)
&\leq \frac12\log\left(4\pi e \left(k(1-p)\left(1-\frac{k}{n}(1-p)\right)\right)\right)-\frac12\log\!\left(\frac{2\pi e}{C} p(1-p)k\right)\\
& = \frac12\log\left(\frac{c'\left(1-\frac{k}{n}(1-p)\right)}{p}\right).
\end{align*}
Therefore, by Fano's inequality,
any decoder achieving vanishing error probability must satisfy
\[
m \;\ge\; \frac{k\log(n/k)}{\log\left(c'\left(\frac{1-k/n+kp/n}{p}\right)\right)}.
\]
for some constant $c'$.
Next, we will argue that the assumption $\E_{S_i}\!\left[\frac{1}{\max\{S_i, 1\}} \right]\leq \frac{C}{k}$ holds when $A$ has i.i.d  $\Ber(1/2)$ entries.
\begin{align*}
\E_{S_i}\!\left[\frac{1}{\max\{S_i, 1\}} \right]&\leq \Pr(S_i\leq k/4) + \frac{4}{k}\\
& \stackrel{(a)}{\leq} e^{-ck} + \frac{4}{k}\\
&\stackrel{(b)}{\leq} \frac{5}{k}.
\end{align*} where a holds because of Chernoff bound for some absolute constant $c$ and $(b)$ holds for large $k$.
\end{proof}
\subsubsection{Linear estimator achievability}
\begin{proof}
Write $D := \supp(x^*)$, $|D|=k$, and for each $i\in[n]$ define the score
\[
S_i := (A^\top y)_i
= \sum_{t=1}^m a_{t,i}\,y_t,
\qquad
y_t = \sum_{s\in D} a_{t,s} z_{t,s}
\]
For $i\in D$,
\[
\mathbb{E}[S_i] = (1-p)m\,\frac{k+1}{4}.
\]
For $j\notin D$,
\[
\mathbb{E}[S_j] = (1-p)m\,\frac{k}{4}.
\]
Hence, for every $i\in D$ and $j\notin D$,
\[
\mathbb{E}[S_i - S_j] = (1-p)\frac{m}{4}.
\]
Fix $i\in D$ and $j\notin D$ and write
\[
S_i - S_j = \sum_{t=1}^m \delta_t,
\qquad
\delta_t := a_{t,i}y_t - a_{t,j}y_t.
\]
Using $y_t = \sum_{s\in D} a_{t,s} z_{t,s}$, we decompose
\[
\delta_t 
= a_{t,i}^2z_{t,i} 
  + \sum_{s\in D\setminus\{i\}} a_{t,i}a_{t,s}z_{t,s}
  - \sum_{s\in D} a_{t,j}a_{t,s}z_{t,s},
\]
Therefore,
\[
\|\delta_t - \mathbb{E}\delta_t\|^2_\text{vp}
\;\le\;
\sigma_{\text{op}}^2(\Ber(\frac{1-p}{2})) + (2k-1) \sigma_{\text{op}}^2(\Ber(\frac{1-p}{4}))
=\frac{p}{2\ln(\frac{1+p}{1-p})} + (2k-1)\frac{1+p}{4\ln(\frac{3+p}{1-p})}.
\]
Let $\mu := \mathbb{E}[S_i - S_j] = (1-p)m/4$.  
Since the $\delta_t$ are independent and subgaussian with the above proxy, the standard subgaussian tail bound (Lemma 1.3 in \cite{rigollet_chapter_nodate}) gives
\[
\Pr(S_i < S_j)
= \Pr\Bigl(\sum_{t=1}^m (\delta_t - \mathbb{E}\delta_t) \le -\mu\Bigr)
\le 
\exp\Bigl(
  - \frac{\mu^2}{2 \sum_{t=1}^m \|\delta_t - \mathbb{E}\delta_t\|^2_\text{vp}}
\Bigr)
\le
\exp\Bigl(
  - \frac{(1-p)^2 m}{32(\frac{p}{2\ln(\frac{1+p}{1-p})} + (2k-1)\frac{1+p}{4\ln(\frac{3+p}{1-p})})}
\Bigr).
\]
An error occurs if $S_i \le S_j$ for some $i\in D$, $j\notin D$.  
By the union bound,
\[
\Pr(\text{Error})
\;\le\;
\sum_{i\in D}\sum_{j\notin D} \Pr(S_i < S_j)
\;\le\;
k(n-k)
\Pr(S_i < S_j).
\]
Therefore, to recover $x^*$ with vanishing error probability, we need 
\[
m\ge \Bigl(
\frac{16p}{(1-p)^2 \ln\frac{1+p}{1-p}}
+
\frac{8(2k-1)(1+p)}{(1-p)^2 \ln\frac{3+p}{1-p}}\Bigr) \log\bigl(k(n-k)\bigr).
\]
\end{proof}
% =========================
\section{Future Directions}
\label{sec:future}
Several directions remain open for future work.
First, while our analysis focuses on a simple correlation-based linear estimator and a least squares benchmark, it would be interesting to design and analyze other computationally efficient algorithms that further narrow the gap to information-theoretic limits.
In particular, iterative schemes that leverage quantitative measurements may achieve improved noise robustness while retaining polynomial-time complexity.

Second, our results cover additive Gaussian noise and Z-channel noise, but many practical settings involve more complex or heterogeneous corruption mechanisms.
Extending the current framework to other noise models would provide a more comprehensive understanding of the robustness limits of noisy quantitative group testing.
% \textcolor{blue}{TODO: Future Work}
\section*{Appendix}
%%%%%%%%%%%%%%%%%%%%%%
\begin{lemma}\label{lemma_N_choose_M}
\[{N\choose M}\le \bigl(\frac{eN}{M}\bigr)^M\]
\end{lemma}
\begin{proof}
    See Lemma 11 in \cite{feige_quantitative_2020} (Appendix A) 
\end{proof}
%%%%%%%%%%%%%%%%%%%%%%%%%
\begin{lemma}\label{lemma_2N_choose_N}
\[{2N\choose N}\le \frac{e\cdot4^N}{\pi \sqrt{2N}}\]
\end{lemma}
\begin{proof}
    See Lemma 12 in \cite{feige_quantitative_2020} (Appendix A)
\end{proof}
%%%%%%%%%%%%%%%%%%%%%%%%
\begin{lemma}\label{lemma_exp}
If $0\le b\le l$, then
    \[4^{-l}{2l \choose l+b}\le \frac{e}{\pi  \sqrt{2l}}\cdot \exp\!\left(-\frac{b^2}{2l}\right)\]
\end{lemma}
\begin{proof}
We begin with
\[
\frac{{2l \choose l+b}}{{2l \choose l}}
= \prod_{i=1}^b \frac{l - i + 1}{l + i}
= \prod_{i=1}^b \left(1 - \frac{2i - 1}{l + i}\right).
\]
Using the inequality $1 - x \le e^{-x}$ for all real $x$, we obtain
\[
\frac{{2l \choose l+b}}{{2l \choose l}}
\le \exp\!\left(-\sum_{i=1}^b \frac{2i - 1}{l + i}\right).
\]
Since $l + i \le 2l$ for all $i \le b$, it follows that
\[
\sum_{i=1}^b \frac{2i - 1}{l + i}
\ge \frac{1}{2l}\sum_{i=1}^b (2i - 1)
= \frac{b^2}{2l}.
\]
By Lemma~\ref{lemma_2N_choose_N}, 
\[4^{-l}{2l \choose l+b}\le 4^{-l}{2l \choose l} \exp\!\left(-\frac{b^2}{2l}\right)
\le \frac{e}{\pi  \sqrt{2l}}\cdot \exp\!\left(-\frac{b^2}{2l}\right)\]
\end{proof}
% \begin{lemma}\label{lemma_4}
% Given vector $V=[v_1,...,v_n]^T$ where each $v_i$ is iid $\Ber(1/2)$, and given a $k$-sparse binary vector $x$, then there exists an absolute constant $C'$ such that $\E\left[\frac{|V|}{Vx} \mathbbm{1}(V>0)\right]\leq C'\cdot\frac{n}{k}$. 
% \end{lemma}
% \paragraph{Acknowledgments.} 
\bibliographystyle{alpha}
\bibliography{references}

\end{document}